%% file: spiked_wigner_short_proof_final.tex
\documentclass[11pt,twoside]{article}

\usepackage{fullpage}

\input{math_macros}

\newcommand{\KL}{\mathsf{KL}}
\newcommand{\RS}{\mathsf{RS}}

\newcommand{\x}{\textup{\texttt{x}}}

\begin{document}

\title{\Large{\bf{Estimation in the spiked Wigner model: \\  A short proof of the replica formula}}}

\author{Ahmed El Alaoui\thanks{Department of EECS, UC Berkeley, CA. Email: elalaoui@berkeley.edu}
\and
Florent Krzakala\thanks{Laboratoire de Physique Statistique, CNRS, PSL Universit\'es \& Ecole Normale Sup\'erieure, Sorbonne Universit\'es et Universit\'e Pierre \& Marie Curie, Paris, France.}
}

\date{}
\maketitle

\vspace{-.3in} 
\begin{abstract} 
We consider the problem of estimating a rank-one perturbation of a Wigner matrix in a setting of low signal-to-noise ratio. This serves as a simple model for principal component analysis in high dimensions. The mutual information per variable between the spike and the observed matrix, or equivalently, the normalized Kullback-Leibler divergence between the planted and null models are known to converge to the so-called {\em replica-symmetric} formula, the properties of which determine the fundamental limits of estimation in this model.  We provide in this note a short and transparent proof of this formula, based on simple executions of Gaussian interpolations and standard concentration-of-measure arguments.  
The \emph{Franz-Parisi potential}, that is, the free entropy at a fixed overlap, plays an important role in our proof.
Furthermore, our proof can be generalized straightforwardly to spiked tensor models of even order.
\end{abstract}

\section{Introduction}
\label{sxn:intro}
Extracting low-rank information from a data matrix corrupted with noise is a fundamental statistical task. 
Spiked random matrix models have attracted considerable attention in statistics, probability and machine learning as rich testbeds for theoretical investigation on this problem~\cite{johnstone2001distribution,peche2014ICM,peche2006largest,baik2005phase}. A basic such model is the {\em spiked Wigner model}  in which one observes a rank-one deformation of a Wigner matrix ${\bm W}$:  
\begin{equation}\label{spiked_wigner_model}
\mtx{Y} = \sqrt{\frac{\lambda}{N}} \vct{x}^* \vct{x}^{*\top} + \mtx{W},
\end{equation}
where $W_{ij} = W_{ji} \sim \normal(0,1)$ and $W_{ii} \sim \normal(0,\sigma^2)$ are independent for all $1 \le i \le j \le N$. The {\em spike} vector $\vct{x}^* \in \R^N$ represents the signal to be recovered, and $\lambda \ge 0$ plays the role of a Signal-to-Noise Ratio (SNR) parameter. The entries of $\vct{x}^*$ come i.i.d.\ from a (Borel) prior $P_{\x}$ on $\R$ with bounded support, so that the scaling in the above model puts the problem in a high-noise regime where only partial recovery of the spike is possible. A basic statistical question about this model is {\em for what values of the SNR $\lambda$ is it possible to estimate the spike $\vct{x}^*$ with non-trivial accuracy?} Spectral methods, or more precisely, estimation using the top eigenvector of $\mtx{Y}$, are know to succeed above a \emph{spectral threshold} and fail below~\cite{benaych2011eigenvalues}.
Since the posterior mean is the estimator with minimal mean squared error, this question boils down to the study of the posterior distribution of $\vct{x}^*$ given $\mtx{Y}$,  which by Bayes' rule, can be written as 
\begin{equation}\label{posterior}
\rmd\P_{\lambda} (\vct{x}| \mtx{Y}) = \frac{ e^{- H(\vct{x})} \rmd P_{\x}^{\otimes N}(\vct{x})}{\int  e^{- H(\vct{x})}  \rmd P_{\x}^{\otimes N}(\vct{x})},
\end{equation} 
where $H$ is the (random) Hamiltonian
\begin{align}\label{hamiltonian}
-H(\vct{x}) &:= \sum_{i < j}  \sqrt{\frac{\lambda}{N}}Y_{ij}x_ix_j -\frac{\lambda}{2N} x_i^2x_j^2 \\
&=  \sum_{i < j}  \sqrt{\frac{\lambda}{N}}W_{ij}x_ix_j + \frac{\lambda}{N}x_ix_jx_i^*x_j^* -\frac{\lambda}{2N} x_i^2x_j^2.\nonumber
\end{align}
Let us define the free entropy\footnote{The term ``free energy" is also used, although the physics convention requires to put a minus sign in front of the expression in this case.} of the model as the expected log-partition function (i.e., normalizing constant) of the posterior $\P_{\lambda}(\cdot | \mtx{Y})$:
\begin{equation}\label{free_energy}
F_{N} = \frac{1}{N} \E \log~ \int e^{- H(\vct{x})} \rmd P_{\x}^{\otimes N}(\vct{x}),
\end{equation}
By heuristically analyzing an approximate message passing (AMP) algorithm for this problem, Lesieur et al.~\cite{lesieur2015phase} derived an asymptotic---so-called \emph{replica-symmetric} ($\RS$)---formula for the above quantity. This formula is defined as follows:   
for $r \in \Rp$, let 
\[\psi(r) := \E_{x^*,z} \log \int \exp\left(\sqrt{r}zx + rxx^* - \frac{r}{2} x^2\right) \rmd P_{\x}(x),\]
where $z \sim\normal(0,1)$ and $x^* \sim P_{\x}$ are mutually independent. 
Define the $\RS$ potential
\[F(\lambda, q) := \psi(\lambda q) -\frac{\lambda q^2}{4}.\]
The conjectured limit of~\eqref{free_energy} is the $\RS$ formula
\[\phi_{\RS}(\lambda) :=  \sup_{q \ge 0} F(\lambda,q).\]
This conjecture was then proved shortly after in a series of papers~\cite{krzakala2016mutual,barbier2016mutual,deshpande2016asymptotic,lelarge2016fundamental} (see also~\cite{korada2009exact}):
\begin{theorem}\label{RS_formula}
For all $\lambda \ge 0$,
\[\lim_{N \to \infty} F_N = \phi_{\RS}(\lambda).\]
\end{theorem}
The above statement contains precious statistical information. It can be written in at least two other equivalent ways, in terms of the mutual information between $\vct{x}^*$ and $\mtx{Y}$: 
 \[\lim_{N \to \infty} ~\frac{1}{N}I(\mtx{Y},\vct{x}^*)~=~\frac{\lambda}{4}\big(\E_{P_{\x}}[X^2]\big)^2 - \phi_{\RS}(\lambda),\]
 or, denoting by $\P_{\lambda}$ the probability distribution of the matrix $\mtx{Y}$ as per~\eqref{spiked_wigner_model}, in terms of the Kullback-Liebler divergence between $\P_{\lambda}$ and $\P_{0}$: 
 \[\lim_{N \to \infty} ~\frac{1}{N}D_{\KL}(\P_{\lambda},\P_{0}) ~=~ \phi_{\RS}(\lambda).\]
 Furthermore, the point $q^*(\lambda)$ achieving the maximum in the $\RS$ formula (which can be shown to be unique and finite for almost every $\lambda$) can be interpreted as the best overlap any estimator $\widehat{\theta}(\mtx{Y})$ can have with the spike $\vct{x}^*$. Indeed, the overlap of a draw $\vct{x}$ from the posterior $\P_{\lambda} (\cdot | \mtx{Y})$ with $\vct{x}^*$ concentrates about $q^*(\lambda)$. See \cite{barbier2016mutual,lelarge2016fundamental,alaoui2017finite} for various forms of this statement.  
 
\section{Comment on the existing proofs} 
The proof of the lower bound $\liminf F_N \ge \phi_{\RS}(\lambda)$
 relies on an application of Guerra's interpolation method~\cite{guerra2001sum,guerra2002thermodynamic}, and is fairly short and transparent. (See~\cite{krzakala2016mutual}.) 
 Available proofs of the converse bound $\limsup F_N \le \phi_{\RS}(\lambda)$ (as well as overlap concentration) are on the other hand highly involved. 
Barbier et al.\ \cite{barbier2016mutual} and Deshpande et al.\ \cite{deshpande2016asymptotic} adopt an algorithmic approach: they analyze an approximate message passing procedure and show that the produced estimator asymptotically achieves an overlap of $q^*(\lambda)$ with the spike. Thus the posterior mean, being the optimal estimator, must also achieve the same overlap. This allows to prove overlap convergence and thus show the converse bound.      
A difficulty one has to overcome with this method is that AMP (and supposedly any other algorithm) may fail to achieve the optimal overlap in the presence of first-order phase transitions, which traps the algorithm in a bad local optimum of the $\RS$ potential. \emph{Spatial coupling}, an idea from coding theory, is used in~\cite{barbier2016mutual} to overcome this problem. Lelarge and Miolane~\cite{lelarge2016fundamental} on the other hand use the Aizenman-Sims-Starr scheme~\cite{aizenman2003extended}, a relative of the cavity method developed within spin-glass theory, to prove the upper bound. Barbier and Macris~\cite{barbier2017stochastic} prove the upper bound via a \emph{adaptive} version of the interpolation method that proceeds via a sequence of intermediate interpolation steps.
Recently, the optimal rate of convergence and constant order corrections to the $\RS$ formula were proved in~\cite{alaoui2017finite} using a rigorous incarnation of the cavity method due to Talagrand~\cite{talagrand2011mean1}. However, all the current approaches (perhaps to a lesser extent for~\cite{barbier2017stochastic}) require the execution of long and technical arguments. 

In this note, we show that the upper bound in Theorem 1 admits a fairly simple proof based on the same interpolation idea that yielded the lower bound, combined with an application of the Laplace method and concentration of measure. The main idea is to consider a version of the free entropy~\eqref{free_energy} of a subsystem of configurations $\vct{x}$ having a \emph{fixed} overlap with the spike $\vct{x^*}$. We then proceed by applying the Guerra bound and optimize over this fixed overlap (which is a free parameter) to obtain an upper bound in the form of a saddle (max-min) formula. A small extra effort is needed to show that this last formula is another representation of the $\RS$ formula. The idea of restricting the overlap dates back to Franz and Parisi~\cite{franz1995recipes,franz1998effective} who introduced it in order to study the relaxation properties of dynamics in spin-glass models. The free entropy at fixed overlap bears the name of the \emph{Franz-Parisi potential}. Our proof thus hinges on a upper bound on this potential, which is may be of independent interest.  We first start by presenting the proof of the lower bound, which is a starting point for our argument. We present the proof in the case $\sigma=\infty$, i.e., we omit the diagonal terms of $\mtx{Y}$. This is only done to keep the displays concise; recovering the general case is straightforward since the diagonal has vanishing contribution to the overall free entropy. Finally, the method presented here can be easily generalized to all spiked tensor models of even order~\cite{richard2014statistical}, thus recovering the main results of~\cite{lesieur2017statistical}.  
%
\section{Proof of Theorem~\ref{RS_formula}}

Let $t \in [0,1]$ and consider an interpolating Hamiltonian
\begin{align}\label{interpolating_hamiltonian}
-H_t(\vct{x}) &:= \sum_{i < j}  \sqrt{\frac{t\lambda}{N}} W_{ij}x_ix_j + \frac{t\lambda}{N}x_ix_i^*x_jx_j^* -\frac{t\lambda}{2N}x_i^2x_j^2 \nonumber\\
&~~+\sum_{i=1}^N  \sqrt{(1-t)r} z_{i}x_i +  (1-t)r x_ix_i^* -\frac{(1-t)r}{2} x_i^2,
\end{align}
where the $z_i$'s are i.i.d.\ standard Gaussian r.v.'s independent of everything else. For $f : (\R^{N})^{n+1} \mapsto \R$, we define the Gibbs average of $f$ as
 \begin{align}\label{gibbs_average_one}
 \left\langle f(\vct{x}^{(1)},\cdots,\vct{x}^{(n)},\vct{x}^*)\right\rangle_t 
 := \frac{\int f(\vct{x}^{(1)},\cdots,\vct{x}^{(n)},\vct{x}^*) \prod_{l=1}^n e^{- H_t(\vct{x}^{(l)})} \rmd P_{\x}^{\otimes N}(\vct{x}^{(l)})}{\int \prod_{l=1}^n e^{- H_t(\vct{x}^{(l)})}  \rmd P_{\x}^{\otimes N}(\vct{x}^{(l)})}. 
 \end{align}
This is the average of $f$ with respect to the posterior distribution of $n$ copies $\vct{x}^{(1)},\cdots,\vct{x}^{(n)}$ of $\vct{x}^*$ given the augmented set of observations
\begin{align}\label{augmented_observations}
\begin{cases}
Y_{ij} &= \sqrt{\frac{t\lambda}{N}} x^*_ix^*_j + W_{ij}, \quad 1 \le i \le j \le N, \\
y_i &= \sqrt{(1-t)r}x^*_i + z_i, \quad 1 \le i \le N. 
\end{cases}
\end{align}  
The variables $\vct{x}^{(l)}, l=1\cdots,n$ are called {\em replicas}, and are interpreted as random variables independently drawn from the posterior. When $n=1$ we simply write $f(\vct{x},\vct{x}^*)$ instead of $f(\vct{x}^{(1)},\vct{x}^*)$. We shall denote the overlaps between two replicas as follows: for $l,l'=1,\cdots,n,*$, we let
 \[R_{l,l'} := \vct{x}^{(l)} \cdot \vct{x}^{(l')} = \frac{1}{N} \sum_{i=1}^N x_i^{(l)}x_i^{(l')}.\]

 A simple consequence of Bayes' rule is that the $n+1$-tuples $(\vct{x}^{(1)},\cdots,\vct{x}^{(n+1)})$ and $(\vct{x}^{(1)},\cdots,\\ \vct{x}^{(n)},\vct{x}^*)$ have the same law under $\E\langle \cdot \rangle_t$ (see Proposition 16 in~\cite{lelarge2016fundamental}). This bears the name of {\em the Nishimori property} in the spin glass literature~\cite{nishimori2001statistical}.  

\subsection{The lower bound} 
Reproducing the argument of~\cite{krzakala2016mutual}, we prove using Guerra's interpolation~\cite{guerra2001sum} and the Nishimori property that
\[F_N \ge \phi_{\RS}(\lambda) - \frac{K}{N}.\]
We let $r = \lambda q$ in the definition of $H_t$ and let
\[\varphi(t) := \frac{1}{N}\E\log \int e^{-H_t(\vct{x})}  \rmd P_{\x}^{\otimes N}(\vct{x}).\]
 A short calculation based on Gaussian integration by parts shows that
{\small 
\begin{align*}
\varphi'(t) = & -\frac{\lambda}{4} \E\left\langle (R_{1,2}-q)^2\right\rangle_{t} + \frac{\lambda}{4} q^2 + \frac{\lambda}{4N^2} \sum_{i=1}^N \E\left\langle {x_i^{(1)}}^2{x_i^{(2)}}^2\right\rangle_{t} \\
&+\frac{\lambda}{2} \E\left\langle (R_{1,*}- q)^2\right\rangle_{t} - \frac{\lambda}{2}q^2 - \frac{\lambda}{2N^2} \sum_{i=1}^N \E\left\langle {x_i}^2{x_i^{*}}^2\right\rangle_{t},
\end{align*}
}
By the Nishimori property, the expressions involving the pairs $(\vct{x}, \vct{x}^*)$ on the one hand and $(\vct{x}^{(1)},\vct{x}^{(2)})$ on the other in the brackets are equal. We then obtain
\[\varphi'(t) = \frac{\lambda}{4} \E\left\langle (R_{1,*}-q)^2\right\rangle_{t} -\frac{\lambda}{4} q^2 - \frac{\lambda}{4N} \E\left\langle {x_N}^2{x_N^{*}}^2\right\rangle_{t}.\]  
Observe that the last term is $\bigo(1/N)$ since the variables $x_N$ are bounded. Moreover, the first term is always non-negative so we obtain
\[\varphi'(t) \ge -\frac{\lambda}{4} q^2 - \frac{K}{N}.\]
Since $\varphi(1) =F_N$ and $\varphi(0) = \psi(\lambda q)$, integrating over $t$, we obtain for all $q \ge0$,
$F_N \ge F(\lambda, q) - \frac{K}{N},$
and this yields the lower bound.

\subsection{The upper bound} 
We prove the converse bound
\[F_N \le \phi_{\RS}(\lambda) + \bigo\Big(\frac{\log N}{\sqrt{N}}\Big).\]
We introduce the Franz-Parisi potential~\cite{franz1995recipes,franz1998effective}. For $\vct{x}^* \in \R^N$ fixed, $m \in \R$ and $\epsilon>0$ we define 
\begin{equation*}\label{free_energy_fixed_overlap}
\Phi_{\epsilon}(m,\vct{x}^*) := \frac{1}{N}\E\log \int \indi\{ R_{1,*} \in [m,m+\epsilon)\} e^{-H(\vct{x})}  \rmd P_{\x}^{\otimes N}(\vct{x}),
\end{equation*}
where the expectation is over $\mtx{W}$.
This is the free entropy of a subsystem of configurations having an overlap close to a fixed value $m$ with a planted signal $\vct{x}^*$. It is clear that $\E_{\vct{x}^*}\Phi_{\epsilon}(m,\vct{x}^*) \le F_N$. We will argue via the Laplace method and concentration of measure that $\sup_{m \in \R} \E_{\vct{x}^*}\Phi_{\epsilon}(m,\vct{x}^*) \approx F_N$, then use Guerra's interpolation to \emph{upper} bound $\Phi_{\epsilon}(m,\vct{x}^*)$ (notice that this method yielded a \emph{lower} bound on $F_N$ due to the Nishimori property).  
Let us define a bit of more notation. For $r \in \Rp, s \in \R$, let 
\[\widehat{\psi}(r,s) := \E_{z} \log \int \exp\left(\sqrt{r}zx + sx - \frac{r}{2} x^2\right) \rmd P_{\x}(x),\]
where $z \sim\normal(0,1)$, and $\widebar{\psi}(r,s) =\E_{x^*}\widehat{\psi}(r,sx^*)$ where $x^* \sim P_{\x}$. 
Moreover, let
\[\widehat{F}(\lambda, m,q,\vct{x}^*) := \frac{1}{N}\sum_{i=1}^N\widehat{\psi}(\lambda q,\lambda m x^*_i)  -\frac{\lambda m^2}{2} +\frac{\lambda q^2}{4},\]
and similarly define $\widebar{F}(\lambda, m,q) = \E_{\vct{x}^*}\widehat{F}(\lambda, m,q,\vct{x}^*) =\widebar{\psi}(\lambda q,\lambda m)  -\frac{\lambda m^2}{2} +\frac{\lambda q^2}{4}$. 

\begin{proposition}\label{laplace_method}
There exist $K>0$ such that for all $\epsilon>0$, we have
\[F_N \le \E_{\vct{x}^*}\Big[\max_{ l \in \Z , |l| \le K/\epsilon  }\Phi_{\epsilon}(l\epsilon,\vct{x}^*)\Big] +  \frac{\log(K/\epsilon)}{\sqrt{N}}.\]
\end{proposition}
Now we upper bound $\Phi_\epsilon$ in terms of $\widehat{F}$:
\begin{proposition}[Interpolation upper bound] \label{fixed_overlap_upper_bound}
There exist $K>0$ depending on $\lambda \ge 0$ such that for all $m \in \R$ and $\epsilon>0$ we have
\[\Phi_{\epsilon}(m,\vct{x^*}) \le \inf_{q \ge 0} \widehat{F}(\lambda, m, q,\vct{x^*}) +  \frac{\lambda}{2} \epsilon^2 + \frac{K}{N}. \]
\end{proposition}

\textbf{Remark:} This simple upper bound on the Franz-Parisi potential---which may be of independent interest---can be straightforwardly generalized to spiked tensor models of even order. Indeed, as will be apparent from the proof in the present matrix case, a crucial step in obtaining the inequality is the positivity of a certain hard-to-control remainder term\footnote{We note that the adaptive interpolation method of Barbier and Macris~\cite{barbier2017stochastic} is able to bypass this issue of positivity of the remainder term along the interpolation path, as long as this interpolation ``stays on the Nishimori line", i.e., corresponds to an inference problem for every $t$ (this is however not true in the case of the FP potential.) They are thus able to compute the free entropy of (asymmetric) spiked tensor models of odd order. See~\cite{barbier2017layered,barbier2017phase}.}. Tensor models of even order enjoy a convexity property that ensures the positivity of this remainder.

From Propositions~\ref{laplace_method} and~\ref{fixed_overlap_upper_bound}, an upper bound on $F_N$ in the form of a saddle formula begins to emerge. For a fixed $m \in \R$ let $\bar{q} = \bar{q}(\lambda,m)$ be any minimizer of $q \mapsto \widebar{F}(\lambda, m,q)$ on $\Rp$. (By differentiating $\widebar{F}$, we can check that $\bar{q}$ is bounded uniformly in $m$.) Then we have
\begin{equation}\label{upper_bound_1}
F_N \le \E_{\vct{x}^*} \hspace{-.1cm}\Big[\max_{m=l\epsilon \atop |l|\le K/\epsilon} \widehat{F}(\lambda,m,\bar{q}(\lambda,m),\vct{x}^*)\Big] +  \frac{\lambda}{2} \epsilon^2 +  \frac{\log(K/\epsilon)}{\sqrt{N}}.
\end{equation}
At this point we need to push the expectation inside the supremum. This will be done using a concentration argument.
\begin{lemma}\label{concentration_planted}
 There exists $K > 0$ such that for all $\lambda \ge 0$, $m\in \R$, $q \ge 0$ and $t \ge 0$, 
 \[\Pr_{\vct{x}^*}\left(\abs{\widehat{F}(\lambda,m,q,\vct{x}^*) - \widebar{F}(\lambda,m,q)} \ge t\right) \le 2e^{-\frac{Nt^2}{\lambda^2K(|m|+q)^2}}.\] 
\end{lemma}
It is a routine computation to deduce from Lemma~\ref{concentration_planted} (and boundedness of both $m$ and $q$) that the expected supremum is bounded by the supremum of the expectation plus a small term (a similar argument is given in the proof of Proposition~\ref{laplace_method}):  
\begin{align*}
\E\hspace{-.1cm}\sup_{m=l\epsilon \atop |l|\le K/\epsilon} \hspace{-.1cm}\widehat{F}(\lambda,m,\bar{q}(\lambda,m),\vct{x}^*) &\le \hspace{-.1cm}\sup_{m=l\epsilon \atop |l|\le K/\epsilon} \hspace{-.1cm}\widebar{F}(\lambda,m,\bar{q}(\lambda,m))+\delta,
\end{align*}
where $\delta = K\log(K/\epsilon)/\sqrt{N}$.
Since $\bar{q}$ is a minimizer of $\widebar{F}$, it follows from~\eqref{upper_bound_1} that
\begin{equation}\label{upper_bound_2}
F_N \le \sup_{m \in \R} \inf_{q \ge 0} \widebar{F}(\lambda,m,q) +  \frac{\lambda}{2} \epsilon^2 +  \frac{K\log(K/\epsilon)}{\sqrt{N}}.
\end{equation}
We now let $\epsilon = N^{-1/4}$, and conclude by noticing that the above saddle formula is another expression for $\phi_\RS$:
\begin{proposition}\label{other_expression}
$\phi_{\RS}(\lambda) =  \sup_{m \in \R} \inf_{q \ge 0} \widebar{F}(\lambda,m,q).$
\end{proposition}
\begin{proof}
One inequality follows from~\eqref{upper_bound_2} and the lower bound $F_N \ge \phi_{\RS}(\lambda) - o_N(1)$. For the converse inequality, we notice that for all $m \in \R$
\[\inf_{q \ge 0} \widebar{F}(\lambda, m, q) \le \widebar{F}(\lambda, m, |m|) = \widebar{\psi}(\lambda |m|, \lambda m) - \frac{\lambda}{4} |m|^2.\]
 Now we use the fact that the function $\widebar{\psi}$ is largest when its second argument is positive:  
\begin{lemma}
\label{asymmetry_lower_bound}
For all $r \ge 0$ we have 
$\widebar{\psi}(r,-r) \le \widebar{\psi}(r,r).$
\end{lemma}
This implies 
$\inf_{q \ge 0} \widebar{F}(\lambda, m, q) \le F(\lambda, |m|) - \frac{\lambda}{4} |m|^2.$
Taking the supremum over $m$ yields the converse bound.
\end{proof}

\noindent\begin{proofof}{Proposition~\ref{laplace_method}}
Let $\epsilon>0$. Since the prior $P_{\x}$ has bounded support, we can grid the set of the overlap values $R_{1,*}$ by $2K/\epsilon$ many intervals of size $\epsilon$ for some $K >0$. This allows the following discretization, where $l$ runs over the finite range $\{-K/\epsilon,\cdots,K/\epsilon\}$: 
\begin{align}
F_{N} &= \frac{1}{N} \E\log \sum_{l} \int \indi\{R_{1,*} \in [l\epsilon, (l+1)\epsilon)\} e^{-H(\vct{x})}  \rmd P_{\x}^{\otimes N}(\vct{x})\nonumber\\
&\le \frac{1}{N} \E\log \frac{2K}{\epsilon}\max_{l} \int \indi\{R_{1,*} \in [l\epsilon, (l+1)\epsilon)\} e^{-H(\vct{x})}  \rmd P_{\x}^{\otimes N}(\vct{x})\nonumber\\
&= \frac{1}{N} \E \max_{l} \log\int \indi\{R_{1,*} \in [l\epsilon, (l+1)\epsilon)\} e^{-H(\vct{x})}  \rmd P_{\x}^{\otimes N}(\vct{x}) + \frac{\log (2K/\epsilon)}{N}.\label{crude_upperBound} 
\end{align}
In the above, $\E$ is w.r.t.\ both $\mtx{W}$ and $\vct{x}^*$. We use concentration of measure to push the expectation over $\mtx{W}$ to the left of the maximum. Let 
\[Z_l := \int   \indi\{R_{1,*} \in [l\epsilon, (l+1)\epsilon)\} e^{-H(\vct{x})}  \rmd P_{\x}^{\otimes N}(\vct{x}).\]
We show that each term $X_l= \frac{1}{N}\log Z_l$ concentrates about its expectation (in the randomness of $\mtx{W}$). Let $\E'$ denote the expectation w.r.t.\ $\mtx{W}$. 
\begin{lemma}\label{sub_gaussian_awkward}
There exists a constant $K >0$ such that for all $\gamma \ge 0$ and all $l$,
\[\E' e^{\gamma(X_l-\E'[X_l])} \le \frac{K \gamma}{\sqrt{N}}e^{K\gamma^2/N}. \]
\end{lemma}
Therefore, the expectation of the maximum concentrates as well: 
\begin{align*} 
\E'\max_{l}(X_l-\E'[X_l])  &\le \frac{1}{\gamma} \log \E' \exp\left(\gamma\max_{l}(X_l-\E'[X_l])\right)\\
&= \frac{1}{\gamma} \log \E' \max_{l}e^{\gamma(X_l-\E'[X_l])}\\
&\le \frac{1}{\gamma} \log \E' \sum_{l}e^{\gamma(X_l-\E'[X_l])}\\
&\le \frac{1}{\gamma} \log \left(\frac{2K}{\epsilon}\frac{\gamma K}{\sqrt{N}}e^{\gamma^2K/N}\right)\\
&= \frac{\log (2K/\epsilon)}{\gamma} + \frac{1}{\gamma}\log\frac{\gamma K}{\sqrt{N}} + \frac{\gamma K}{N}.
\end{align*}
We set $\gamma = \sqrt{N}$ and obtain 
\[\E'\max_{l} (X_l-\E'[X_l]) \le \frac{\log(K/\epsilon)}{\sqrt{N}}.\]
Therefore, plugging the above estimates into~\eqref{crude_upperBound}, we obtain
\begin{align*}
F_N &\le \E_{\vct{x}^*}\max_{l}  \E' X_l + \frac{\log (K/\epsilon)}{\sqrt{N}} +  \frac{\log (K/\epsilon)}{N}\\
&\le  \E_{\vct{x}^*}\max_l \Phi_{\epsilon}(l\epsilon,\vct{x}^*)  + 2\frac{\log (K/\epsilon)}{\sqrt{N}}.
\end{align*}
\end{proofof}

\noindent\begin{proofof}{Proposition~\ref{fixed_overlap_upper_bound}}
Let $t \in [0,1]$ and consider a slightly modified interpolating Hamiltonian that has two parameters $r  = \lambda q \ge 0$ and $s = \lambda m \in \R$:
\begin{align}\label{interpolating_hamiltonian_2}
-H_t(\vct{x}) &:= \sum_{i < j} \sqrt{\frac{t\lambda}{N}} W_{ij}x_ix_j + \frac{t\lambda}{N} x_ix_i^*x_jx_j^*  -\frac{t\lambda}{2N} x_i^2x_j^2 \\
&~~+ \sum_{i=1}^N  \sqrt{(1-t)r} z_{i}x_i +  (1-t)s x_ix_i^*  -\frac{(1-t)r}{2} x_i^2,\nonumber
\end{align}
where the $z_i$'s are i.i.d.\ standard Gaussian r.v.'s independent of everything else. 
Let 
\[\varphi(t) := \frac{1}{N}\E\log \int  \indi\{R_{1,*} \in [m,m+\epsilon)\}  e^{-H_t(\vct{x})}  \rmd P_{\x}^{\otimes N}(\vct{x}),\]
where $\E$ is over the Gaussian disorder $\mtx{W}$ and $\vct{z}$ ($\vct{x}^*$ is fixed). Let $\langle \cdot \rangle_t$ be the corresponding Gibbs average, similarly to~\eqref{gibbs_average_one}.  By differentiation and Gaussian integration by parts,
\begin{align*}
\varphi'(t) = & -\frac{\lambda}{4} \E\left\langle (R_{1,2}-q)^2\right\rangle_{t} + \frac{\lambda}{4} q^2 + \frac{\lambda}{4N^2} \sum_{i=1}^N \E\left\langle (x_i^{(1)} x_i^{(2)})^2\right\rangle_{t} \\
&+\frac{\lambda}{2} \E\left\langle (R_{1,*}- m)^2\right\rangle_{t} - \frac{\lambda}{2}m^2 - \frac{\lambda}{2N^2} \sum_{i=1}^N \E\left\langle (x_ix_i^{*})^2\right\rangle_{t},
\end{align*}
Notice that by the overlap restriction, $\E\left\langle (R_{1,*}- m)^2\right\rangle_{t} \le \epsilon^2$. Moreover, the last terms in the first and second lines in the above are of order $1/N$ since the variables $x_i$ are bounded. Next, since $ \E\left\langle (R_{1,2}-q)^2\right\rangle_{t}$ has non-negative sign (this is a crucial fact), we can ignore it and obtain an upper bound:
\[\varphi'(t) \le - \frac{\lambda}{2}m^2 + \frac{\lambda}{4} q^2 + \frac{\lambda}{2} \epsilon^2 + \frac{ K}{N}.\]
Integrating over $t$, we obtain
\begin{equation*}
\Phi_{\epsilon}(m,\vct{x}^*) \le - \frac{\lambda}{2}m^2 + \frac{\lambda }{4}q^2 + \frac{\lambda}{2} \epsilon^2 + \varphi(0) + \frac{ K}{N}.
\end{equation*}
Now we use a trivial upper bound on $\varphi(0)$: 
\begin{align*}%
\varphi(0) &=  \frac{1}{N}\E\log \int \indi\{R_{1,*} \in [m,m+\epsilon)\} e^{-H_0(\vct{x})}  \rmd P_{\x}^{\otimes N}(\vct{x})\\
&\le \frac{1}{N}\E\log \int e^{-H_0(\vct{x})}  \rmd P_{\x}^{\otimes N}(\vct{x})\\
&= \frac{1}{N}\sum_{i=1}^N\widehat{\psi}(\lambda q,\lambda m x^*_i).
\end{align*}
Hence,
\[\Phi_{\epsilon}(m,\vct{x}^*) \le \widehat{F}(\lambda,m,q,\vct{x}^*)  + \frac{\lambda}{2} \epsilon^2 + \frac{ K}{N}.\]
\end{proofof}

\vspace{.4cm}
\noindent\begin{proofof}{lemma~\ref{concentration_planted}}
The random part of $\widehat{F}(\lambda,m,\bar{q},\vct{x}^*)$ is the average of i.i.d.\ terms $\widehat{\psi}(\lambda q,\lambda mx^*_i)$. Since $\abs{\partial_s\widehat{\psi}(r,sx^*)} \le K^2$, $\abs{\partial_r\widehat{\psi}(r,sx^*)} \le K^2/2$ and $\widehat{\psi}(0,0)=0$, where $K$ is a bound on the support of $P_{\x}$, we have $\abs{\widehat{\psi}(r,sx^*)} \le K^2(r/2+|s|)$. For bounded $r$ and $s$, the claim follows from concentration of the average of i.i.d.\ bounded r.v.'s. 
\end{proofof}

\vspace{.4cm}
\noindent\begin{proofof}{Lemma~\ref{sub_gaussian_awkward}}
 We notice that $X_l$ seen as a function of $\mtx{W}$ is Lipschitz with constant $K\sqrt{\frac{\lambda}{N}}$. By Gaussian concentration of Lipschitz functions (the Tsirelson-Ibragimov-Sudakov inequality~\cite{boucheron2013concentration}), there exist a constant $K$ depending only on $\lambda$ such that for all $t \ge 0$, 
\[\Pr\left(X_l - \E' X_l  \ge t\right) \le e^{-Nt^2/K}.\]
Then we conclude by means of the identity
\[\E' e^{\gamma(X_l-\E'[X_l])} = \gamma \int_{-\infty}^{+\infty} \Pr(X_l-\E'[X_l] \ge t)~e^{\gamma t} \rmd t,\]
and integrate the tail.
\end{proofof}

\vspace{.4cm}
\begin{proofof}{Lemma~\ref{asymmetry_lower_bound}}
Let $\nu = P_{\x}$, and let $\mu$ be the symmetric part of $P_{\x}$, i.e., $\mu(A) = (P_{\x}(A)+P_{\x}(-A))/2$ for all Borel $A \subseteq \R$. Observe that $\nu$ is absolutely continuous with respect to $\mu$.
The argument relies on a linearly interpolating between the two measures $\mu$ and $\nu$. Let $t \in [0,1]$ and let $\rho_t = (1-t)\mu + t\nu$. Further, let $r >0$ be fixed, and
\[\varphi^\pm(t) := \E_{z} \hspace{-.1cm}\int \left(\log \int \exp\left(\sqrt{r}zx \pm r xx^* - \frac{r}{2} x^2\right) \rmd \rho_t(x)\right) \rmd \rho_t(x^*),\]
where $z \sim \normal(0,1)$. Now let $\phi(t) = \varphi^+(t) - \varphi^-(t)$.
We have $\phi(1) = \widebar{\psi}(r,r) - \widebar{\psi}(r,-r)$ on the one hand, and since $\mu$ is a symmetric distribution, $\phi(0) = 0$ on the other. 
We will show that $\phi$ is a convex increasing function on the interval $[0,1]$. Then we deduce that $\phi(1) \ge 0$. First, we have
\begin{align*}
\frac{\rmd}{\rmd t} \varphi^+(t) &= \E_{z} \int \log \int e^{\sqrt{r}zx + r xx^* - \frac{r}{2} x^2} \rmd \rho_t(x) ~\rmd (\nu - \mu)(x^*)\\
&~~~+\E_{z} \int \frac{\int e^{\sqrt{r}zx + r xx^* - \frac{r}{2} x^2}\rmd(\nu - \mu)(x)}{\int e^{\sqrt{r}zx + r xx^* - \frac{r}{2} x^2}\rmd \rho_t(x)} ~\rmd \rho_t(x^*),
\end{align*}
and
\begin{align*}
\frac{\rmd^2}{\rmd t^2} \varphi^+(t) &= 2\E_{z} \int \frac{\int e^{\sqrt{r}zx + r xx^* - \frac{r}{2} x^2}\rmd(\nu - \mu)(x)}{\int e^{\sqrt{r}zx + r xx^* - \frac{r}{2} x^2}\rmd \rho_t(x)} \rmd (\nu - \mu)(x^*)\\
& -2\E_{z} \int \left(\frac{\int e^{\sqrt{r}zx + r xx^* - \frac{r}{2} x^2}\rmd(\nu - \mu)(x)}{\int e^{\sqrt{r}zx + r xx^* - \frac{r}{2} x^2}\rmd \rho_t(x)}\right)^2 \rmd \rho_t(x^*).
\end{align*}

Similar expressions holds for $\varphi^-$ where $x^*$ is replaced by $-x^*$ inside the exponentials. We see from the expression of the first derivative at $t=0$ that ${\varphi^+}'(0) = {\varphi^-}'(0)$. This is because $\rho_0 = \mu$ is symmetric about the origin, so a sign change (of $x$ for the first term, and $x^*$ for the second term) does not affect the value of the integrals. Hence $\phi'(0) = 0$. 
Now, we focus on the second derivative. Observe that since $\mu$ is the symmetric part of $\nu$, $\nu - \mu$ is anti-symmetric. This implies that the first term in the expression of the second derivative changes sign under a sign change in $x^*$ and keeps the same modulus. As for the second term, a sign change in $x^*$ induces integration against $\rmd \rho_t(-x^*)$. Hence we can write the difference $(\varphi^+-\varphi^-)''$ as  
\begin{align*}
\phi''(t) &= 4 \E_{z} \int \frac{\int e^{\sqrt{r}zx + r xx^* - \frac{r}{2} x^2}\rmd(\nu - \mu)(x)}{\int e^{\sqrt{r}zx + r xx^* - \frac{r}{2} x^2}\rmd \rho_t(x)} ~\rmd (\nu - \mu)(x^*)\\ 
&~~-2\E_{z} \int \left(\frac{\int e^{\sqrt{r}zx + r xx^* - \frac{r}{2} x^2}\rmd(\nu - \mu)(x)}{\int e^{\sqrt{r}zx + r xx^* - \frac{r}{2} x^2}\rmd \rho_t(x)}\right)^2(\rmd \rho_t(x^*)-\rmd \rho_t(-x^*)).
\end{align*}

For any Borel $A$, we have $\rho_t(A)- \rho_t(-A) = (1-t)(\mu(A) - \mu(-A)) + t(\nu(A)-\nu(-A)) = 2t(\nu - \mu)(A)$. Therefore the second term in the above expression becomes
\[-4t\E_{z} \int \left(\frac{\int e^{\sqrt{r}zx + r xx^* - \frac{r}{2} x^2}\rmd(\nu - \mu)(x)}{\int e^{\sqrt{r}zx + r xx^* - \frac{r}{2} x^2}\rmd \rho_t(x)}\right)^2 \rmd(\nu - \mu)(x^*).\]
Since both $\mu$ and $\nu$ are absolutely continuous with respect to $\rho_t$ for all $0 \le t < 1$ we can write 
\[\phi''(t) = 4 \E_{z,x^*} \left\langle \frac{\rmd(\nu - \mu)}{\rmd\rho_t}(x) \frac{\rmd(\nu - \mu)}{\rmd\rho_t}(x^*) \right\rangle - 4 t\E_{z,x^*} \left\langle \frac{\rmd(\nu - \mu)}{\rmd\rho_t}(x)\right\rangle^2,\]
where the Gibbs average is with respect to the posterior of $x$ given $z,x^*$ under the Gaussian channel $y = \sqrt{r}x^* + z$, and the expectation is under $x^* \sim \rho_t$ and $z \sim \normal(0,1)$. By the Nishimori property, we simplify the above expression to
\[ \phi''(t) = 4(1-t)\E \left[\left\langle \frac{\rmd(\nu - \mu)}{\rmd\rho_t}(x)\right\rangle^2\right] \ge 0,\]
where the expression is valid for all $0\le t <1$. From here we see that the function $\phi$ is convex on $[0,1]$. Since $\phi(0) = \phi'(0) = 0$, $\phi$ is also increasing on $[0,1]$.    
\end{proofof}

\vspace{5mm}
\noindent\textbf{Acknowledgments.}
Florent Krzakala acknowledges funding from the ERC under the European Union 7th Framework Programme Grant
Agreement 307087-SPARCS. 

\begin{small}

\bibliographystyle{alpha}

\bibliography{phase_transitions}

\end{small}

\end{document}

%% file: math_macros.tex
\usepackage{amsthm, amsmath, amssymb, dsfont, mathtools,bm}

\theoremstyle{plain}
\newtheorem{theorem}{Theorem}

\newtheorem{lemma}[theorem]{Lemma}

\newtheorem{proposition}[theorem]{Proposition}
 \newenvironment{proofof}[1]{{\bf {\em Proof of #1.}}}{\hfill \rule{2mm}{2mm} 
 }
  
\newlength{\widebarargwidth}
\newlength{\widebarargheight}
\newlength{\widebarargdepth}
\DeclareRobustCommand{\widebar}[1]{%
  \settowidth{\widebarargwidth}{\ensuremath{#1}}%
  \settoheight{\widebarargheight}{\ensuremath{#1}}%
  \settodepth{\widebarargdepth}{\ensuremath{#1}}%
  \addtolength{\widebarargwidth}{-0.3\widebarargheight}%
  \addtolength{\widebarargwidth}{-0.3\widebarargdepth}%
  \makebox[0pt][l]{\hspace{0.3\widebarargheight}%
    \hspace{0.3\widebarargdepth}%
    \addtolength{\widebarargheight}{0.3ex}%
    \rule[\widebarargheight]{0.95\widebarargwidth}{0.1ex}}%
  {#1}}
  
\makeatletter
\long\def\@makecaption#1#2{
        \vskip 0.8ex
        \setbox\@tempboxa\hbox{\small {\bf #1:} #2}
        \parindent 1.5em  
        \dimen0=\hsize
        \advance\dimen0 by -3em
        \ifdim \wd\@tempboxa >\dimen0
                \hbox to \hsize{
                        \parindent 0em
                        \hfil 
                        \parbox{\dimen0}{\def\baselinestretch{0.96}\small
                                {\bf #1.} #2
                                } 
                        \hfil}
        \else \hbox to \hsize{\hfil \box\@tempboxa \hfil}
        \fi
        }
\makeatother

\newcommand{\normal}{\ensuremath{\mathcal{N}}}

\newcommand{\bigo}{\ensuremath{\mathcal{O}}}

\renewcommand{\P}{\operatorname{\mathbb{P}}}
\newcommand{\E}{\operatorname{\mathbb{E}}}

\newcommand{\Z}{\mathbb{Z}}

\newcommand{\R}{\mathbb{R}}

\newcommand{\Rp}{\mathbb{R}_{+}}

\newcommand{\indi}{\mathds{1}}

\newcommand{\rmd}{\mathrm{d}}

\newcommand{\abs}[1]{\left|#1\right|}
\newcommand{\vct}[1]{\bm{#1}}
\newcommand{\mtx}[1]{\bm{#1}}

\def\BC{\begin{center}}
\def\EC{\end{center}}
\def\BIT{\begin{itemize}}
\def\EIT{\end{itemize}}
\def\BET{\begin{enumerate}}
\def\EET{\end{enumerate}}
\def\BEQ{\begin{equation}}
\def\EEQ{\end{equation}}

\long\def\comment#1{}
